%% file: paper.tex
\documentclass[]{article}

\include{commands}

\title{Correlation Clustering Algorithm for Dynamic Complete Signed Graphs: An Index-based Approach}
\author{Ali Shakiba}
\date{\footnotesize{Department of Computer Science, Vali-e-Asr University of Rafsanjan, Rafsanjan, Iran. \\ \footnotesize{\texttt{ali.shakiba@vru.ac.ir;a.shakiba.iran@gmail.com}}}}

\begin{document}
	
	\maketitle
	
	\begin{abstract}
		\input{abstract}
	\end{abstract}

	\section{Introduction}\label{sec:introduction}
		\input{intro}

	\section{Preliminaries}\label{sec:preliminary}
		\input{prelim}

	\section{Proposed Method}\label{sec:proposed}	
		\input{proposed}

	\section{Experiments}\label{sec:experiments}
		\input{experim}

	\section{Conclusion}\label{sec:conclusion}
		\input{conclusion}


\end{document}

%% file: commands.tex
\usepackage{amsmath}
\usepackage{amsthm}
\usepackage{amssymb}
\usepackage{amsfonts}
\usepackage{graphicx}
\usepackage{fullpage}
\usepackage{algorithm}
\usepackage{algpseudocode}
\usepackage{xspace}
\usepackage{multirow}
\usepackage{hyperref}
\usepackage{caption}
\usepackage{subcaption}

\usepackage{todonotes}

\newtheorem{thm}{Theorem}
\newtheorem{lem}{Lemma}

\newtheorem{obs}{Observation}
\newtheorem{coro}{Corollary}
\theoremstyle{definition}
\newtheorem{defn}{Definition}

\theoremstyle{remark}
\newtheorem{rem}{Remark}

\newcommand{\set}[1]{\left\{#1\right\}}
\newcommand{\card}[1]{\left|#1\right|}
\newcommand{\bigoh}[1]{\mathcal{O}\left(#1\right)}
\newcommand{\epsagreeformula}[3]{\frac{\card{N_{#1}(#2) \Delta N_{#1}(#3)}}{\max\set{\card{N_{#1}(#2)},\card{N_{#1}(#3)}}}}
\newcommand{\nonagreement}[3]{\textsc{NonAgreement}_{#1}\left(#2,#3\right)}

\newcommand{\agreecnt}{\textsc{AgreeCnt}\xspace}
\newcommand{\epsagree}{$\varepsilon$-agreement\xspace}
\newcommand{\epslight}{$\varepsilon$-light\xspace}
\newcommand{\epshigh}{$\varepsilon$-heavy\xspace}

\newcommand{\np}{\textbf{NP}}
\newcommand{\nao}[1]{\mathbb{NAO}\left(#1\right)}
\newcommand{\lth}[2]{\mathbb{LTH}_{#1}\left(#2\right)}

\DeclareMathOperator{\cost}{cost}

\newcommand{\epsschedule}{$\varepsilon$-\textsc{Schedule}\xspace}
\newcommand{\epssched}{\varepsilon\textsc{-Schedule}\xspace}

%% file: abstract.tex
In this paper, we reduce the complexity of approximating the correlation clustering problem from $\bigoh{m\times\left( 2+ \alpha (G) \right)+n}$ to $\bigoh{m+n}$ for any given value of $\varepsilon$ for a complete signed graph with $n$ vertices and $m$ positive edges where $\alpha(G)$ is the arboricity of the graph. Our approach gives the same output as the original algorithm and makes it possible to implement the algorithm in a full dynamic setting where edge sign flipping and vertex addition/removal are allowed. Constructing this index costs $\bigoh{m}$ memory and $\bigoh{m\times\alpha(G)}$ time. We also studied the structural properties of the non-agreement measure used in the approximation algorithm. The theoretical results are accompanied by a full set of experiments concerning seven real-world graphs. These results shows superiority of our index-based algorithm to the non-index one by a decrease of \%34 in time on average. 

\noindent\textbf{Keywords:} Correlation clustering $\cdot$ Dynamic graphs $\cdot$ Online Algorithms

%% file: intro.tex
Clustering is one of the most studied problems in machine learning with various applications in analyzing and visualizing large datasets. There are various models and technique to obtain a partition of elements, such that elements belonging to different partitions are dissimilar to each other and the elements in the same partition are very similar to each other. 

The problem of correlation clustering, introduced in \cite{bbc04-correlation-clustering}, is known to be an \np-hard problem for the disagree minimization. Therefore, several different approximation solutions based on its IP formulation exist in the literature. Recently, the idea of a 2-approximation algorithm in \cite{bbc04-correlation-clustering} is extended in \cite{clmnpt21-ICML-correlation-clustering-in-constant-many-parallel-rounds} for constructing a $\bigoh{1}$-approximation algorithm. The experiments in \cite{clmnpt21-ICML-correlation-clustering-in-constant-many-parallel-rounds} show acceptable performance for this algorithm in practice, although its theoretical guarantee can be too high, e.g. $1\,442$ for $\beta=\lambda=\frac{1}{36}$. In \cite{clmp22-ICML-online-consistent-correlation-clustering}, this algorithm is extended to an online setting where just vertex additions are allowed, and whenever a new vertex is added, it reveals all its positively signed edges. Shakiba in \cite{shakiba2022online} studied the effect of vertex addition/removal and edge sign flipping in the underlying graph to the final clustering result, in order to make the algorithm suitable for dynamic graphs. However, one bottleneck in this way is computing the values of \textsc{NonAgreement} among the edges and identifying the $\varepsilon$-lightness of vertices. The current paper proposes a novel indexing scheme to remedy this and make the algorithm efficient, not just in terms of dynamic graphs, but for even dynamic hyper-parameter $\varepsilon$. Our proposed method, in comparison with the online method of \cite{clmp22-ICML-online-consistent-correlation-clustering} is that we allow a full dynamic setting, i.e. vertex addition/removal and edge's sign flipping. It is known that any online algorithm for the correlation clustering problem has at least $\Omega(n)$-approximation ratio \cite{mathieu2010online}. Note that the underlying algorithm used in the current paper is consistent, as is shown via experimental results \cite{clmp22-ICML-online-consistent-correlation-clustering}.

The rest of the paper is organized as follows: In Section \ref{sec:contrib}, we highlight our contributions. This is followed by a reminding some basic algorithms and results in Section \ref{sec:preliminary}. Then, we introduce the novel indexing structure in Section \ref{sec:indexing:structure} and show how it can be employed to enhance the running-time of the approximate correlation clustering algorithm. Then, we show how to maintain the proposed indices in a full dynamic settings in Section \ref{sec:index:maintenance}. In Section \ref{sec:experiments}, we give an extensive experiments which accompanies the theoretical results and show the effectiveness of the proposed indexing structure. Finally, a conclusion is drawn. 

\subsection{Our Contribution}\label{sec:contrib}
	In this paper, we simply ask 
	\begin{quote}
		``How can one reduce the time to approximate a correlation clustering of the input graph \cite{clmnpt21-ICML-correlation-clustering-in-constant-many-parallel-rounds} for varying values of $\varepsilon$?''
	\end{quote}
	We also ask 
	\begin{quote}
		``How can we make the solution to the first question an online solution for dynamic graphs?''
	\end{quote}
	Our answer to the first question is devising a novel indexing-structure which is constructed based on the structural properties of the approximation algorithm and its \textsc{NonAgreement} measure. As our experiments in Section \ref{sec:experiments} show, the proposed method enhanced the total running-time of querying the clustering for about $\%34$ on average for seven real-world datasets. Then, we make this structure online to work with dynamic graphs based on theoretical results in \cite{shakiba2022online}. The construction of the index itself is highly parallelizable, up to the number of the vertices in the input graph. The idea for parallelization is simple: construct each $\nao{v}$ in the $\nao{G}$ with a separate parallel thread. 
	We also study the intrinsic structures in the \textsc{NonAgreement} measure, to bake more efficient algorithms for index-maintenance due to updates to the underlying graph. 
	More precisely, we show that using the proposed index structure, we can find a correlation clustering for a graph for any given value of $\varepsilon$ in time $\bigoh{m+n}$, compared to the $\bigoh{m\times\left(2+\alpha(G)\right)+n}$ time for the CC. The pre-processing time of the ICC would be $\bigoh{m\times \alpha(G)}$ with $\bigoh{m}$ space complexity.

%% file: prelim.tex
Let $G=(V,E)$ be a complete undirected signed graph with $\card{V}=n$ vertices. The set of edges $E$ is naturally partitioned into positive and negative signed edges, $E^+$ and $E^-$, respectively. Then, we use $m$ to denote $\card{E^+}$. 
The correlation clustering problem is defined as 
\begin{equation}
	\label{eq:min:disagree:corr:clustering}
	\cost(\mathcal{C}) = \sum_{\substack{\set{u,v}\in E^+\\u\in C_i,v\in C_j,i\neq j}} 1 + \sum_{\substack{\set{u,v}\in E^-\\u,v\in C_i}} 1,
\end{equation}
where $\mathcal{C}=\set{C_1,\ldots,C_\ell}$ is a clustering. Note that this is the min-disagree variant of the problem. 

The constant factor approximation algorithm of \cite{clmnpt21-ICML-correlation-clustering-in-constant-many-parallel-rounds} is based on two main quantities: (1) $\varepsilon$-agreement of a positively signed edge $\set{u,v}$, i.e. $u$ and $v$ are in $\varepsilon$-agreement if and only if $\nonagreement{G}{u}{v} = \epsagreeformula{G}{u}{v} < \varepsilon$, and (2) $\varepsilon$-lightness, where a vertex $u$ is said to be $\varepsilon$-light if $\frac{\agreecnt_{G^+}(u)}{\card{N_{G^+}(u)}} < \varepsilon$ where $\agreecnt_{G^+}(u)=\card{\set{w\in V|u \text{ and } v \text{ are in } \varepsilon\text{-agreement}}}$. Note that a vertex which is not $\varepsilon$-light is called $\varepsilon$-heavy. This is a $2+\frac{4}{\varepsilon}+\frac{1}{\varepsilon^2}$-approximation algorithm, as is shown in \cite{clmp22-ICML-online-consistent-correlation-clustering}. This algorithm is described in Algorithm \ref{alg:correlation:clustering}, which we will refer to the CC algorithm, for short. 
\begin{algorithm}[h]
	\caption{\textsc{CorrelationClustering($G$)} \cite{clmnpt21-ICML-correlation-clustering-in-constant-many-parallel-rounds}}
	\label{alg:correlation:clustering}
	\begin{algorithmic}[1]
		\Procedure{CorrelationClustering}{$G,\varepsilon$}
		\State Let $G^+=G[E^+]$ where $E^+$ is the set of edges whose sign is $+$
		\State Discard all edges whose endpoints are not in \epsagree
		\State Discard all edges between two $\varepsilon$-light vertices
		\State Let $\widetilde{G^+}$ be the \emph{sparsified} graph $G^+$ after performing previous two operations
		\State Let $\mathcal{C}$ be the collection of connected components in $\widetilde{G^+}$
		\State \Return $\mathcal{C}$ as the output clustering
		\EndProcedure
	\end{algorithmic}
\end{algorithm}

Shakiba in \cite{shakiba2022online} studied theoretical foundation of the CC algorithm in a full dynamic setting. The following result is a summary of Table 1, Corollary 1, and Theorem 4 in \cite{shakiba2022online}.
\begin{thm}
	\label{thm:locality:change:shakiba}
	Suppose the sign of an edge $u=\set{u,v}$ is flipped. Then, the non-agreement and $\varepsilon$-lightness of vertices other than the ones whose distance to either $u$ and $v$ is more than two would not change. 
\end{thm}

The arboricity of the graph $G$ is the minimum number of edge-disjoint spanning forests into which $G$ can be decomposed. The following lemma for arboricity is useful in bounding the number of operations. 
\begin{lem}{Lemma 2 in \cite{cn85-arboricity-subgraph-listing-algorithms}}
	\label{lem:arboricity}
	Suppose the graph $G=(V,E)$ has $n$ vertices with $m$ edges. Then, 
	\begin{equation}
		\sum_{\set{u,v}\in E} \min \set{\deg_{G}(u), \deg_{G}(v)} \leq 2a(G) \times m.
	\end{equation}
\end{lem}

%% file: proposed.tex
In this section, we describe our novel indexing structure. This structure allows dynamic queries of the correlation clustering with varying values of $\varepsilon$ for dynamic graphs. The proposed algorithm which uses the indexing structure would be called ICC, or indexed-based correlation clustering. 

\subsection{Indexing structure}\label{sec:indexing:structure}

For an edge $e=\set{u,v}$ with positive sign, we define its \emph{\epsagree distance} as $\nonagreement{G^+}{u}{v}$. Intuitively, this is the infimum of the values $\varepsilon$ which the nodes $u$ and $v$ are not in \epsagree. Let define the set  $\mathcal{E}=\set{\nonagreement{G^+}{u}{v}\vert e=\set{u,v}\in E^+}$. Without loss of generality, let $\mathcal{E}=\set{\varepsilon_0,\ldots,\varepsilon_{\ell-1}}$ with the ordering $\min\mathcal{E}=\varepsilon_0<\varepsilon_1<\cdots<\varepsilon_{\ell-1}=\max\mathcal{E}$. For a fixed value of $\varepsilon$, let $G^+_\varepsilon=(V,E^+_\varepsilon)$ where $E^+_\varepsilon=\set{e=\set{u,v}\in E^+\vert \nonagreement{G^+}{u}{v}<\varepsilon}$. 
\begin{obs}
	For all $\varepsilon \leq \varepsilon_0$, $G^+_\varepsilon$ is the null graph, i.e. a graph on all nodes without any edges. Moreover, for all $\varepsilon > \varepsilon_\ell$, $G^+_\varepsilon=G^+$. 
\end{obs}
Next, we introduce the key ingredient to our indexing structure, called $\mathbb{NAO}$. 
\begin{defn}[\textsc{NonAgreement} Node Ordering]
	\label{def:nao}
	The \epsagree ordering for each node $v\in V$, denoted by $\nao{v}$, is defined as an ordered subset of vertices in $G$ where:
	\begin{enumerate}
		\item node $u\in V$ appears in the ordering $\nao{v}$ if and only if $e=\set{u,v}$ is a positive edge in $G$.
		\item for each two distinct vertices $u,w\in V$ which appear in $\nao{v}$, 
			\begin{equation}
				\nonagreement{G^+}{v}{u} < \nonagreement{G^+}{v}{w},
			\end{equation}
			implies $u$ appears before $w$. 
		\item for each node $u\in \nao{v}$, its \epsagree distance is also stored with that node.
	\end{enumerate}
	The \textsc{NonAgreement} node ordering of the graph $G$ is defined as $\nao{G}=\set{(v,\nao{v})\vert v\in V}$. 
\end{defn}
In other words, the $\nao{v}$ is a sorted array of neighboring nodes of $v$ in $G^+$ in their \epsagree distance value. An example \textsc{NonAgreement} node ordering for all vertices in a sample graph is illustrated in Figure \ref{fig:non:agree:order}. 
\begin{figure}[h]
	\centering
	\includegraphics[width=.4\textwidth]{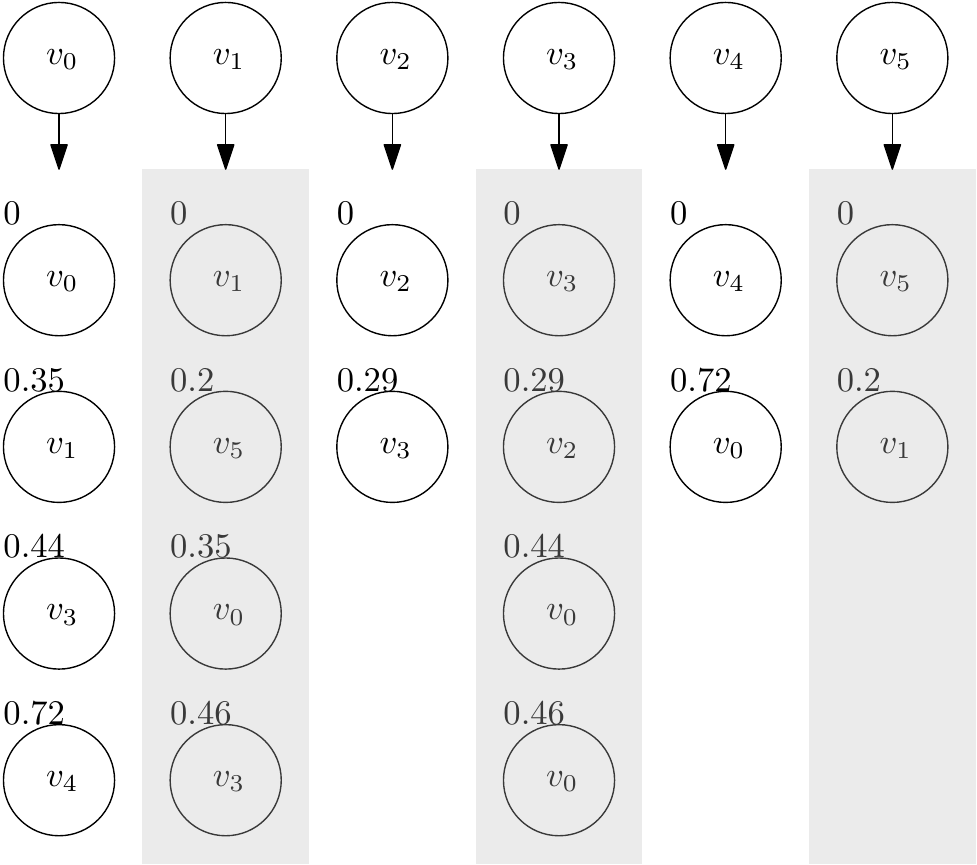}
	\caption{An illustrative example of $\nao{v}$ for an example graph.}
	\label{fig:non:agree:order}
\end{figure}
The space and construction time complexities of the $\mathbb{NAO}(G)$ are investigated in the next two lemmas.
\begin{lem}
	\label{lem:nao:space:complexity}
	The \textsc{NonAgreement} node ordering for a graph $G$, $\nao{G}$, can be represented in $\bigoh{m}$ memory. 
\end{lem}
\begin{proof}
	The number of nodes inside $\nao{v}$ equals to the $\deg_{G^+}(v)+1$, for all vertices in $N_{G^+}(v)$ as well as the vertex $v$ itself. Note that it is not required to explicitly store the vertex $v$ itself in the ordering. Cumulatively, the total size required for representing $\nao{v}$ is $2\times m$ entries. 
\end{proof}
\begin{lem}
	\label{lem:nao:construction:time}
	The time complexity to construct the $\nao{v}$ for all vertices $v\in V$ is $\bigoh{m\times\left(\alpha(G) + \lg m\right)}$ where $\alpha(G)$ is the arboricity of the graph $G$. 
\end{lem}
\begin{proof}
	To compute the value of $\nonagreement{G^+}{u}{v}$ for all edges $e=\set{u,v}\in E^+$, one needs to compute $\card{N_{G^+}(u) \Delta N_{G^+}(v)}$. This requires $\bigoh{\deg_{G^+}(u) + \deg_{G^+}(v)}$ operations to compute the symmetric difference, given the adjacency lists of $u$ and $v$ are sorted. However, we can compute  their intersection and use it to compute the \textsc{NonAgreement} as follows
	\begin{equation}
		\nonagreement{G^+}{u}{v}=\frac{\deg_{G^+}(u) + \deg_{G^+}(v) - 2\times\card{N_{G^+}(u) \cap N_{G^+}(v)}}{\max\set{N_{G^+}(u),N_{G^+}(v)}}.
	\end{equation}
	Hence, the total time required to compute the values of \textsc{NonAgreement} is equal to $$\sum_{\set{u,v}\in E^+} \min \set{\deg_{G^+}(u),\deg_{G^+}(v)},$$ which is known to be bounded by $2\alpha(G)\times m$ (Lemma \ref{lem:arboricity}). Moreover, each $\nao{v}$ is of length $1+\deg_{G^+}(v)$ which requires sorting by the value of their \textsc{NonAgreement} in $\bigoh{\deg_{G^+}(v)\lg\left(\deg_{G^+}(v)\right)}$, which accumulates to $\bigoh{m\lg m}$. 
\end{proof}
The correlation clustering corresponding to a given value of $\varepsilon$ and a graph $G$ is the set of connected components of the graph $\widetilde{G^+}$. Given access to $\nao{G}=\set{\nao{v}\vert v\in V(G)}$, one can respond to the following queries: (1) Is the vertex $v$ an \epshigh vertex or an \epslight? (2) Are the endpoints of an edge $\set{u,v}$ in \epsagree? As we will show in this section, both of these questions can be answered in $\bigoh{1}$ time using the $\mathbb{NAO}$ structure.
For a vertex $v\in V$, its \epslight threshold is defined as $\lth{\varepsilon}{v}=\left\lceil\varepsilon\deg_{G^+}(v)\right\rceil+1$. 
\begin{lem}
	\label{lem:lightness:wrt:nao}
	For an $\varepsilon$ and a vertex $v\in V$, $v$ is \epshigh if and only if the \epsagree distance of the $\lth{\varepsilon}{v}$-th smallest vertex in $\nao{v}$ is less than $\varepsilon$. Otherwise, it is \epslight. 
\end{lem}
\begin{proof}
	The vertices $u$ and $v$ would be in $\varepsilon$-agreement if and only if $\nonagreement{G^+}{u}{v}<\varepsilon$. Whenever the \epsagree distance of the $\lth{\varepsilon}{v}$-th smallest vertex in $\nao{v}$ is less than $\varepsilon$, it means that there are at least $\left\lceil\varepsilon\deg_{G^+}(v)\right\rceil+1$ vertices, including the $v$ itself, in $\varepsilon$-agreement with $v$. This is equivalent the \epshigh{ness} of the vertex $v$. On the other hand, if the \epsagree distance of the $\lth{\varepsilon}{v}$-th smallest vertex in $\nao{v}$ is greater than or equal to $\varepsilon$, this is equivalent to that the number of vertices in \epsagree with $v$ is less than $\varepsilon \times \deg_{G^+}(v)$, i.e. $v$ is \epslight.
\end{proof}
\begin{lem}
	\label{lem:lightness:wrt:nao:time}
	Given access to $\nao{v}$, for all values of $\varepsilon$ and any vertex $v \in V$, identifying the $\varepsilon$-lightness of $v$ can be accomplished in $\bigoh{1}$.
\end{lem}
\begin{proof}
	This is implied by Lemma \ref{lem:lightness:wrt:nao}, assuming that the $\nao{v}$ is implemented as an array. 
\end{proof}
Given access to the index $\nao{G}=\set{(v,\nao{v})\vert v\in V}$ for a graph $G$, a query simply asks for a clustering of the graph for a given value of $\varepsilon$. 
\begin{thm}
	\label{thm:cluster:using:nao}
	Given access to the index $\nao{G}$, computing a clustering for a given parameter $\varepsilon$ can be accomplished in $\bigoh{m\left(a(n)+1\right)}$ amortized time, where $a(n)$ is the slowly growing inverse of the single-valued Ackermann's function. 
\end{thm}
\begin{proof}
	Intuitively, we are going to use $\nao{G}$ to construct the graph $\widetilde{G^+}$ and compute its connected components incrementally. More formally, we start by putting each vertex to its isolated set in a disjoin-union structure. Then, we identify \epshigh vertices $\mathcal{H}\subseteq V$, which takes $\bigoh{m}$. Next, consider $\nao{v}$ for $v\in \mathcal{H}$. For each vertex $u\in\nao{v}$ whose key is smaller than $\varepsilon$ and $u\in\mathcal{H}$, we call $\textsc{Union}(u,v)$, which merges the cluster which $u$ and $v$ belong. These operations takes $\bigoh{m\times a(n)}$ amortized time \cite{clrs22-introduction-algorithms} using a Disjoint-Set data structure. 
	The correctness of the query algorithm is implied by the Definition \ref{def:nao} and the Lemmas \ref{lem:lightness:wrt:nao} and \ref{lem:lightness:wrt:nao:time}. 
\end{proof}

Before going any further, we need to state the following results, which intuitively investigate the behavior of the \epsagree and \epslight of edges and vertices through the $\mathcal{E}$.%
Let 
\begin{equation}
	\mathcal{E}=\set{\nonagreement{G^+}{u}{v}\vert\set{u,v}\in E^+}=\set{\varepsilon_0, \varepsilon_1, \ldots, \varepsilon_{\ell-1}}, 
\end{equation}
where $\varepsilon_{i-1} < \varepsilon_{i}$ for $i=1,\ldots,\ell-1$. Moreover, assume that $\varepsilon_\ell>\max\mathcal{E}$. 
\begin{obs}
	\label{obs:min:eps:values}
	For any $\varepsilon\leq\varepsilon_0$, the endpoints of all positive signed edges $\set{u,v}$ are not in $\varepsilon$-agreement. 
\end{obs}
\begin{obs}
	\label{obs:max:eps:values}
	For any $\varepsilon>\varepsilon_{\ell-1}$, the endpoints of all positive signed edges $\set{u,v}$ are in $\varepsilon$-agreement.
\end{obs}
By the Observation \ref{obs:min:eps:values}, the correlation clustering output by the Algorithm \ref{alg:correlation:clustering} for all values of $\varepsilon\leq\varepsilon_0$ would be the collection of singleton vertices, i.e. $\set{\set{v}\vert v\in V}$. However, we cannot say that for $\varepsilon>\varepsilon_{ell-1}$, the output is a single cluster, i.e. the set $V$. Why is that? As $\varepsilon$ increases, the number of edges in \epsagree would increase, however, the number of $\varepsilon$-light vertices would increase, too (By Corollary \ref{cor:monotonicity:result}, which follows soon). Hence, there is a trade-off for choosing a suitable value of $\varepsilon$ to get the minimum number of clusters. We would discuss this issue further in Section \ref{sec:experiments} (Experiments). 
\begin{thm}
	\label{thm:monotonicity:agree}
	Let $\varepsilon < \varepsilon'$ and $u,v\in V$ be two distinct vertices. If $u$ and $v$ are in $\varepsilon$-agreement, then they would be in $\varepsilon'$-agreement, too. Also, if $u$ and $v$ are not in $\varepsilon'$-agreement, then they would not be in $\varepsilon$-agreement. 
\end{thm}
\begin{proof}
	The proof is a direct implication of the \textsc{NonAgreement} definition. Given that $u$ and $v$ are in $\varepsilon$-agreement, we have $\nonagreement{G^+}{u}{v}<\varepsilon<\varepsilon'$, i.e. $u$ and $v$ are in $\varepsilon'$-agreement, too. Similarly, given that $u$ and $v$ are not in $\varepsilon'$-agreement, we have $\nonagreement{G^+}{u}{v}\geq\varepsilon'>\varepsilon$, i.e. $u$ and $v$ are not in $\varepsilon$-agreement, too.
\end{proof}
\begin{thm}
	\label{thm:monotonicity:lightness}
	Let $\varepsilon < \varepsilon'$ and $u\in V$. If $u$ is $\varepsilon$-light, then $u$ would be $\varepsilon'$-light, too. Also, if $u$ is $\varepsilon'$-heavy, then it would be $\varepsilon$-heavy. 
\end{thm}
\begin{proof}
	This is implied by the definition of $\varepsilon$-lightness and Theorem \ref{thm:monotonicity:agree}. 
\end{proof}
Two other important cases are not considered in Theorem \ref{thm:monotonicity:lightness}, i.e. (1) Is it possible that a vertex $u$ be $\varepsilon$-heavy, but becomes $\varepsilon'$-light for some values $\varepsilon<\varepsilon'$?, and (2) Is it possible that a $\varepsilon$-light vertex becomes $\varepsilon'$-heavy for some values of $\varepsilon<\varepsilon'$? The answer to both of these questions is affirmative. %
By Theorems \ref{thm:monotonicity:agree} and \ref{thm:monotonicity:lightness} and the previous discussion, we can state the following corollary. 
\begin{coro}
	\label{cor:monotonicity:result}
	Let $\varepsilon < \varepsilon'$. 
	\begin{enumerate}
		\item The number of positively signed edges whose endpoints are in $\varepsilon'$-agreement is greater than or equal to the number of positive edges whose endpoints are in $\varepsilon$-agreement. In other words, the agreement relation is monotone.
		\item The number of vertices which are $\varepsilon'$-light can be either greater or less than or even equal to the number of $\varepsilon$-light vertices. In other words, the lightness relation is not necessarily monotone.
	\end{enumerate}
\end{coro}

The \textsc{Baseline} idea is to recompute the \textsc{NonAgreement} for each new value of $\varepsilon$, which takes $\bigoh{m\times \alpha(G)}$, as is described in the proof of Lemma \ref{lem:nao:space:complexity}, deciding on the $\varepsilon$-heaviness of the vertices in $G$ in time $\bigoh{n}$ and computing the connected components of the graph $\widetilde{G^+}$ as the output in time $\bigoh{m+n}$. Totally, the time complexity of the \textsc{Baseline} is $\bigoh{m\times\left(2+\alpha(G)\right)+n}$. 

To conclude, using the index structure we invest $\bigoh{m}$ space (Lemma \ref{lem:nao:space:complexity}) and $\bigoh{m\times\left(\alpha(G) + \lg m\right)}$ (Lemma \ref{lem:nao:construction:time}) time to construct the $\nao{G}$ which make it possible to answer each query with varying $\varepsilon$ values in time $\bigoh{m(a(n)+1)}$ (Theorem \ref{thm:cluster:using:nao}). Comparing this with $\bigoh{m\times\left(2+\alpha(G)\right)+n}$ time for the \textsc{Baseline} reveals that our index-based structure makes query times faster for variable values of $\varepsilon$. 

Given access to \textsc{NAO}, the algorithm \ref{alg:correlation:clustering:nao} constructs the graph $\widetilde{G^+}$. Note that the output of the algorithms \ref{alg:correlation:clustering} and \ref{alg:correlation:clustering:nao} are the same. As answering whether a vertex $v$ is \epshigh or the endpoints of an edge $\set{u,v}$ are in \epsagree can be answered in $\bigoh{1}$ time, the running-time of the algorithm \ref{alg:correlation:clustering:nao} would be $\bigoh{\max\set{|V|,|E^+|}}$ for the \textbf{for} loop and $\bigoh{|V_{\widetilde{G^+}}|+|E^+_{\widetilde{G^+}}|}$ to compute the connected components of $\widetilde{G^+}$. Totally, the running-time of the query algorithm in the ICC would be $\bigoh{m+n}$, compared to the $\bigoh{m\times\left(2+\alpha(G)\right)+n}$ time for the CC. 

To model queries over a graph for various values of $\varepsilon$, we use a \epsschedule defined as
\begin{equation}
	\epssched=\set{\varepsilon_0 < \varepsilon_1 < \cdots < \varepsilon_\ell} \subseteq [0,2]. 
\end{equation}
Note that we consider $\epssched$ as a set of strictly increasing real numbers, just for the sake of simplicity in notation. In a real life scenario, any time one can ask for a clustering with any value of $\varepsilon$. 

\begin{algorithm}[h]
	\caption{The index-based correlation clustering algorithm with $\nao{G}$ structure.}%
	\label{alg:correlation:clustering:nao}
	\begin{algorithmic}[1]
		\Procedure{IndexCorrelationClustering}{$G,\nao{G},\varepsilon$}
			\ForAll{$v\in V$}
				\State Use the $\nao(v)$ to identify and remove all the edges $\set{v,u}$ which are in $\varepsilon$-\textsc{NonAgreement}
				\State Use the $\nao{v}$ to identify and remove edges $\set{u,v}$ where $u$ and $v$ are both $\varepsilon$-light
			\EndFor
			\State \Return The connected components of the remaining graph as the clustering output.
		\EndProcedure
	\end{algorithmic}
\end{algorithm}

\subsection{Maintaining the index}\label{sec:index:maintenance}
The index structure introduced in Section \ref{sec:indexing:structure} is used to compute a clustering of a static graph for user-defined dynamic values of $\varepsilon$. In this section, we revise this structure to make it suitable for computing a clustering of a dynamic graph for dynamic values of $\varepsilon$. 

There are three different operations applicable to the underlying graph of the \textsc{CorrelationClustering}: (1) flipping the sign of an edge $e=\set{u,v}$, (2) adding a new vertex $v$, and (3) removing an existing vertex $v$. These operations are considered in Lemmas \ref{lem:edge:sign:flip:update:nao}, \ref{lem:add:vertex:update:nao}, and \ref{lem:remove:vertex:batch}, respectively. 

Shakiba in \cite{shakiba2022online} has shown that flipping the sign of an edge $\set{u,v}$, the \epsagree of edges whose both endpoints are not in the union of the positive neighborhood of its endpoints does not change (Propositions 2 and 4 in \cite{shakiba2022online}). Therefore, we just need to compute the \epsagree for positive edges $\set{x,w}$ where $x,w\in \left( N_{G^+}(u) \cup N_{G^+}(v) \right)$.

\begin{lem}
	\label{lem:edge:sign:flip:update:nao}
	Assuming the sign of an edge $e=\set{u,v}$ is flipped.
	\begin{enumerate}
		\item Assuming the sign of edge was $+$ prior to the flipping, then $u$ and $v$ would be no more in $\varepsilon$-agreement since their edge is now negatively signed. The vertices $v$ and $u$ are removed from the arrays $\nao{u}$ and $\nao{v}$, respectively. Moreover, we need to update the values of the non-agreement for vertices $u$ and $v$ in $\nao{w}$ for all vertices $w \in \nao{u} \cup \nao{v}$. 
		\item Assuming the sign of edge was $-$ prior to the flipping, then their $\nonagreement{G^+}{u}{v}$ is computed and the vertices $v$ and $u$ are added to proper sorted place in $\nao{u}$ and $\nao{v}$, respectively. Moreover, we need to recompute the $\nonagreement{G^+}{x}{w}$ for all $x,w\in\nao{u}\cup\nao{v}$ and update their \epsagree values. 
	\end{enumerate}
	Applying these changes is applicable in $\bigoh{\sum_{\set{x,w}\in E^+ \cap X \times X}\left(\log\deg_{G^+}(x)+\log\deg_{G^+}(w)\right)}$ where $X=\nao{u}\cup\nao{v}$. 
\end{lem}
\begin{proof}
	The correctness of this lemma follows from the discussion just before it. Therefore, we would just give the time analysis. In the first case, finding the vertices in each $\mathbb{NAO}$ is of $\bigoh{\log \deg_{G^+}(u)+\log \deg_{G^+}(v)}$. The removal would be $\bigoh{1}$ amortized time. 
	
	The second case, i.e. flipping the sign from $-$ to $+$, would cost more, as it may require changing all the positive edges with both endpoints the in the set $X$. In the worst case, one may assume that all the \epsagree between the edges with both endpoints inside $X$ is changed. Therefore, we need to update each of them on their corresponding $\mathbb{NAO}$ indices. We need to add vertices $v$ and $u$ to the $\mathbb{NAO}(u)$ and $\mathbb{NAO}(v)$, respectively, based on $\nonagreement{G^+}{u}{v}$. Finding their correct position inside sorted $\mathbb{NAO}$s is possible with $\bigoh{\log \deg_{G^+}(u)+\log \deg_{G^+}(v)}$ comparisons. Again, adding them at their corresponding position would take $\bigoh{1}$ amortized time. For the other positive edges with both endpoints in the set $X$, such as $\set{x,w}$, we might need to update their \epsagree. Without loss of generality, assume that $\deg_{G^+}(x) \leq \deg_{G^+}(w)$. Doing so requires re-computation of $\nonagreement{G^+}{x}{w}$, in $\bigoh{\min\set{\deg_{G^+}(x),\deg_{G^+}(w)}}$, querying the \textsc{NonAgreement} inside $\nao{x}$, in time $\bigoh{\log\deg_{G^+}(x)}$, and then possible updating its position, in both $\nao{x}$ and $\nao{w}$ requires $\bigoh{\log\deg_{G^+}(x)+\log\deg_{G^+}(w)}$. In the worst case, all the elements in $X$ may need update. Therefore, this case can be accomplished in $$\bigoh{\sum_{\set{x,w}\in E^+ \cap X \times X}\left(\log\deg_{G^+}(x)+\log\deg_{G^+}(w)\right)},$$ in the worst-case. 
\end{proof}
Using Lemma \ref{lem:edge:sign:flip:update:nao}, we can state the $\textsc{NAO-FlipEdge}(u,v)$ algorithm (Algorithm \ref{alg:edge:sign:flip:update:nao}) which flips the sign of the edge $\set{u,v}$. The correctness and the running-time of this algorithm follows directly from Lemma \ref{lem:edge:sign:flip:update:nao}. 
\begin{algorithm}[h]
	\caption{The $\textsc{NAO-FlipEdge}(u,v)$ algorithm to apply the effect of flipping the sign of the edge $(u,v)$ in $\nao{G}$.}
	\label{alg:edge:sign:flip:update:nao}
	\begin{algorithmic}[1]
		\Procedure{NAO-FlipEdge}{$u,v$}
			\State Let $\mathcal{A} \gets \nao{u}\cup\nao{v}$.
			\If{Sign of edge $\set{u,v}$ is positive}
				\State Update the graph $G^+$ by removing edge $\set{u,v}$.
				\State Remove the vertices $u$ and $v$ from $\nao{v}$ and $\nao{u}$, respectively.
			\Else \Comment{Sign of edge $\set{u,v}$ is negative}
				\State Update the graph $G^+$ by adding edge $\set{u,v}$.				
				\State Compute $\nonagreement{G^+}{u}{v}$.
				\State Add the vertices $u$ and $v$ to $\nao{v}$ and $\nao{u}$, respectively. %
			\EndIf
			\ForAll{$w\in\mathcal{A}$}
				\State Recompute $\nonagreement{G^+}{u}{w}$ and update $\nao{u}$ and $\nao{w}$.
				\State Recompute $\nonagreement{G^+}{v}{w}$ and update $\nao{v}$ and $\nao{w}$.
			\EndFor
		\EndProcedure
	\end{algorithmic}
\end{algorithm}
There are cases where a set of positive edges $E^+_v$ are also given for a newly added vertex $v$. One can first add the vertex with all negative edges, and afterwards, flip all the edges in $E^+_v$, so they would become positively signed. However, a batch operation would give us higher performance in practice. 
\begin{lem}
	\label{lem:add:vertex:update:nao}
	Assume that a vertex $v$ is added to graph $G$ with new positive signed edges $E^+_v$. Then,
	\begin{enumerate}
		\item If $E^+_v = \emptyset$, then $\nao{G}$ does not need any updates and we just add a new $\nao{v}$ to $\nao{G}$ with $v$ as its only element in constant-time.
		\item Otherwise, let $X=N_{G^+}(v) \cup \left(\cup_{x\in N_{G^+}(v)} N_{G^+}(x)\right)$. We need to compute the $\nonagreement{G^+}{x}{w}$ for each positive edges $\set{x,w}$ with $x,w\in X$ after adding all the new positive edges in $E^+_v$, and possibly update $\nao{x}$ and $\nao{w}$. This would require 
		\begin{equation}
			\bigoh{\sum_{\set{x,w}\in (E^+ \cup E^+_v) \cap X \times X}\left(\log\deg_{G^+}(x)+\log\deg_{G^+}(w)\right)},
		\end{equation}
		operations in the worst case.
	\end{enumerate}
\end{lem}
\begin{proof}
	The correctness would follow immediately from the Lemma \ref{lem:edge:sign:flip:update:nao}. Again, the first case can be accomplished in $\bigoh{1}$ time. For the second case, we need to calculate $\card{X}$ values of \textsc{NonAgreement}, and add them or update them in their corresponding $\mathbb{NAO}$s. This would require 
	\begin{equation}
		\bigoh{\sum_{\set{x,w}\in (E^+ \cup E^+_v) \cap X \times X}\left(\log\deg_{G^+}(x)+\log\deg_{G^+}(w)\right)}, 
	\end{equation}
	by exactly the same discussion as in Lemma \ref{lem:edge:sign:flip:update:nao}. 
\end{proof}
\begin{rem}
	In comparing the time required for Lemmas \ref{lem:edge:sign:flip:update:nao} and \ref{lem:add:vertex:update:nao}, please note that the size of the set $X$ can be much larger in Lemma \ref{lem:add:vertex:update:nao} than the one in Lemma \ref{lem:edge:sign:flip:update:nao}, depending on the size of $E^+_v$ for the new vertex $v$. 
\end{rem}
The $\textsc{NAO-AddVertex}(x,N_x)$ algorithm (Algorithm \ref{alg:add:vertex}) adds a new vertex $x$ to $G^+$ and all its neighboring positively signed edges $N_x$, and updates the $\nao{G}$. The correctness and the running-time of this algorithm follows directly from Lemma \ref{lem:add:vertex:update:nao}. 
\begin{algorithm}[h]
	\caption{The $\textsc{NAO-AddVertex}(x,N_x)$ algorithm to add a new vertex $x$ and all its positive neighbors $N_x$ to $\nao{G}$.}
	\label{alg:add:vertex}
	\begin{algorithmic}[1]
		\Procedure{NAO-AddVertex}{$x,N_x$}
			\State Update graph $G^+$ by adding the new vertex $x$.
			\State Construct a new $\nao{x}$ and add it to $\nao{G}$.
			\State Let $X \gets N_x \cup \left(\cup_{z\in N_x} N_{G^+}(z)\right)$.
			\ForAll{$y \in X$}
				\ForAll{$z \in X\setminus y$}
					\State Update $\nao{y}$ and $\nao{z}$ if the $\nonagreement{G^+}{y}{z}$ changes.
				\EndFor
			\EndFor
		\EndProcedure
	\end{algorithmic}
\end{algorithm}
For the last operation, removing an existing vertex $v$ with a set of positive adjacent edges $E^+_v$, can be accomplished by first flipping the sign of all of its adjacent edges in $E^+_v$ from $+$ to $-$, and afterwards, removing its $\nao{v}$ from $\nao{G}$. Similar to vertex addition, we can do it also in batch-mode, hoping for better performance in practice. The algorithm for the $\textsc{NAO-RemoveVertex}(x)$ is similar to the $\textsc{NAO-AddVertex}(x,N_x)$ algorithm (Algorithm \ref{alg:add:vertex}).
\begin{lem}
	\label{lem:remove:vertex:batch}
		Assume that an existing vertex $v$ is removed from the graph $G$ with the set of adjacent positive signed edges $E^+_v$. Then,
	\begin{enumerate}
		\item If $E^+_v = \emptyset$, then $\nao{v}$ has a single element, itself, and can be easily removed from $\nao{G}$. Nothing else would require a change, so it is of constant-time complexity.
		\item Otherwise, let $X=N_{G^+}(v) \cup \left(\cup_{x\in N_{G^+}(v)} N_{G^+}(x)\right)$. We need to compute the $\nonagreement{G^+}{x}{w}$ for each positive edges $\set{x,w}$ with $x,w\in X$ after removing all the edges in $E^+_v$, and possibly update $\nao{x}$ and $\nao{w}$. In the worst-case, this would require 
		\begin{equation}
			\bigoh{\sum_{\set{x,w}\in (E^+ \cup E^+_v) \cap X \times X}\left(\log\deg_{G^+}(x)+\log\deg_{G^+}(w)\right)},
		\end{equation}
		operations.
	\end{enumerate}
\end{lem}
\begin{proof}
	The correctness of this lemma is implied by the Lemma \ref{lem:edge:sign:flip:update:nao}. The discussion of the running-time is exactly the same as in the proof of Lemma \ref{lem:add:vertex:update:nao}. 
\end{proof}

%% file: experim.tex
To evaluate the proposed method, we used 7 graphs which are formed by user-user interactions. These datasets are described in Table \ref{tbl:datasets} and are accessible through the SNAP \cite{SNAP}. The smallest dataset consists of 22\,470 nodes with 171\,002 edges and the largest one consists of 317\,080 nodes with 1\,049\,866 edges. In all these graphs, we consider the existing edges as positively signed and non-existing edges as negatively signed. 

The distribution of the \textsc{NonAgreement}s for each dataset is illustrated in Figures \ref{fig:non:agree:1} to \ref{fig:non:agree:7}. The distribution of \textsc{NonAgreement}s of the edges in almost all the datasets obeys the normal distribution, except small imperfections in Arxiv ASTRO-PH (Figure \ref{fig:non:agree:1}), Arxiv COND-MAT (Figure \ref{fig:non:agree:3}), and DBLP (Figure \ref{fig:non:agree:7}). 
Moreover, more detailed statistics on these distributions are given in Tables \ref{tbl:nonagree:stats}. One single observation is that the most frequent value of the \textsc{NonAgreement} in all the sample datasets is $1$. Why is that? Without loss of generality, assume that $\deg_{G}(u) \geq \deg_{G}(v)$. Then, $\card{N_G(v)} = 2\card{N_G(u) \cap N_G(v)}$. By the assumption $\card{N_G(u)} \geq 2 \card{N_G(u) \cap N_G(v)}$, i.e. the intersection of the neighborhood of vertices $u$ and $v$ consists of at most half of the neighborhood of $u$. Also, exactly half of the neighborhood of $v$ falls at the intersection of the neighborhood of $u$ and $v$. Intuitively, the vertex $u$ has clustering preference with extra vertices at least the number of vertices which both $u$ and $v$ have clustering preference. Similarly, the vertex $v$ has clustering preference with exactly half extra vertices which both $u$ and $v$ have clustering preference.

For each dataset, we have used a different set of {\epsschedule}s, depending on the distribution of their \textsc{NonAgreement}s. More precisely: (1) we have sorted the values of non-agreements in each dataset in a non-decreasing order, with repetitions. (2) Then, we have selected $21$ distinct values equally spaces of these values. (3) The \epsschedule was set to the selected values in the second step, which appended the value of $0$ at the beginning and the value of $1.99$ to the end if either does not exists. Totally, the number of \epsschedule for each dataset is either $21$, $22$ or $23$. 

An interesting observation is the number of clusters for $\varepsilon\approx 1^-$ in Figures \ref{fig:cluster:1} to \ref{fig:cluster:7}. Note that we use $1^-$ to denote the interval $[1-\epsilon,1]$ for some non-zero constant $\epsilon>0$. When $\varepsilon \approx 1$ but less than $1$, the number of clusters is minimum. As it gets closer to $1$, the number of clusters increases with a much greater descent than the decrease in the number of clusters as it gets close to $1$ from $0$. 

As in Corollary \ref{cor:monotonicity:result}, by increasing the $\varepsilon$, the number of vertices in $\varepsilon$-agreement is a non-decreasing function, which is confirmed by the plots in Figures \ref{fig:non:agree:1} to \ref{fig:non:agree:7} as the number of vertices in $\varepsilon$-non-agreement is given by a non-increasing function. By closer visual inspection of these figures, we can see that the shape of the plot for the number of $\varepsilon$-non-agreement vertices in all these graphs is almost the same, with inflection point around the value of $\varepsilon\approx 1$. This is due to the intrinsic nature of the $\nonagreement{G}{u}{v}$. 

Similarly, the non-monotonicity result stated in Corollary \ref{cor:monotonicity:result} is observed in the same figures for the number of $\varepsilon$-light vertices. By a visual inspection, the trend of the number of $\varepsilon$-light vertices for almost all datasets, except for the Arxiv HEP-TH (Figure \ref{fig:non:agree:4}) and the EU-Email (Figure \ref{fig:non:agree:6}), is the same: the number of $\varepsilon$-light vertices increases as $\varepsilon$ increases up to some point (first interval), then decreases slightly (second interval), and finally increases and would be asymptotically equal to the number of vertices in the graphs (last interval). For Arxiv HEP-TH and EU-Email, we have the same trend, however, the second interval is very small.

\begin{table}[h]
	\centering
	\caption{Description of the datasets.}
	\label{tbl:datasets}
	\begin{tabular}{|l|r|r|}
		\hline 
		\textbf{Dataset} & \textbf{Nodes} & \textbf{Edges} \\ \hline 
		\textit{Arxiv ASTRO-PH} \cite{condensed} & 18\,772 & 198\,110 \\ \hline 		
		\textit{MUSAE-Facebook} \cite{musae:facebook} & 22\,470 & 171\,002 \\ \hline 
		\textit{Arxiv COND-MAT} \cite{condensed} & 23\,133 & 93\,497 \\ \hline 
		\textit{Arxiv HEP-TH} \cite{citations} & 27\,770 & 352\,807 \\ \hline 
		\textit{Enron-Email} \cite{enron} & 36\,692 & 183\,831 \\ \hline 
		\textit{EU-Email} \cite{condensed} & 265\,214 & 420\,045 \\ \hline 
		\textit{DBLP} \cite{dblp} & 317\,080 & 1\,049\,866 \\ \hline 
	\end{tabular}
\end{table}

\begin{table}[h]
	\centering
	\caption{Statistics of \textsc{NonAgreement} in each dataset.}
	\label{tbl:nonagree:stats}
	\begin{tabular}{|l|r|r|r||r|r|}
		\hline 
		\textbf{Dataset} & \textbf{Distinct} & \textbf{Minimum} & \textbf{Maximum} & \multicolumn{2}{c|}{\textbf{Top 2 frequent values}} \\ \hline 
		\textit{Arxiv ASTRO-PH} & 16\,436 & 0.015\,873\,0 & 1.967\,21 & 1 | 9\,664 & 0.5 | 2\,992 \\ \hline 		
		\textit{MUSAE-Facebook} & 12\,988 & 0.031\,746 & 1.978\,02 & 1 | 18\,534 & 1.25 | 2\,346 \\ \hline
		\textit{Arxiv COND-MAT} & 3\,893 & 0.044\,444\,4 & 1.964\,91 & 1 | 12\,018 & 0.5 | 4\,954 \\ \hline 
		\textit{Arxiv HEP-TH} & 23\,285 & 0.086\,956\,5 & 1.976\,19 & 1 | 24\,852 & 1.333\,33 | 4\,910 \\ \hline 
		\textit{Enron-Email} & 20\,273 & 0.090\,909\,1 & 1.954\,55 & 1 | 31\,704 & 0.5 | 6\,796 \\ \hline 
		\textit{EU-Email} & 28\,612 & 0.117\,647 & 1.990\,52 & 1 | 441\,520 & 0.997\,658 | 682 \\ \hline 
		\textit{DBLP} & 11\,611 & 0.013\,793\,1 & 1.981\,13 & 1 | 167\,590 & 0.5 | 67\,238 \\ \hline 
	\end{tabular}
\end{table}

All the algorithms for the naive and the proposed index-based correlation clustering algorithms are implemented in \texttt{C++}\footnote{\href{https://github.com/alishakiba/Correlation-Clustering-Algorithm-for-Dynamic-Complete-Signed-Graphs-An-Index-based-Approach}{https://github.com/alishakiba/Correlation-Clustering-Algorithm-for-Dynamic-Complete-Signed-Graphs-An-Index-based-Approach}} without any parallelization and the experiments are done using an Ubuntu 22.04.1 TLS with an Intel Core i7-10510U CPU @ 1.80GHz with 12 GB of RAM.
The time for running the naive correlation clustering algorithm (Algorithm \ref{alg:correlation:clustering}), denoted here as CC, as well as the time for the index-based correlation clustering algorithm denoted as ICC, is given in Figures \ref{fig:times:1} to \ref{fig:times:7}). Note that the time reported for the ICC in these figures does not include the required time for constructing the $\mathbb{NAO}$s, as they are constructed once and used throughout the \epsschedule. The running-time to read the graph as well as constructing the $\mathbb{NAO}$s is reported in Table \ref{tbl:time:construction:nao} in milliseconds. The CC and ICC algorithms are the same, except that in CC, the non-agreement values of the edges and the $\varepsilon$-lightness of the vertices are computed for each given value of $\varepsilon$, however, in ICC these are computed and stored in the proposed $\mathbb{NAO}$s structure once and used for clustering with respect to different values of $\varepsilon$. As it can be observed in Figures \ref{fig:times:1} to \ref{fig:times:7}, the running time for the ICC, excluding the time to construct the $\mathbb{NAO}$s for once, is largely smaller than the one for CC. On average, our approach for the described \epsschedule lead to $\%25$ decrease in clustering time. This enhancement comes at the cost of pre-computing the $\mathbb{NAO}$s, which costs on average $\%34$ of the time for a single run of CC, which is quite small and makes the ICC efficient in cases where one requires to have multiple clustering for various values of $\varepsilon$.

\begin{table}[h]
	\centering
	\caption{Time to construct the $\mathbb{NAO}$ for each dataset in milliseconds.}
	\label{tbl:time:construction:nao}
	\begin{tabular}{|l|r|r|}
		\hline 
		\textbf{Dataset} & \textbf{Time to construct $\mathbb{NAO}$ (ms)} & \textbf{Time to read graph (ms)} \\ \hline 
		\textit{Arxiv ASTRO-PH} & 1\,677 & 543 \\ \hline 		
		\textit{MUSAE-Facebook} & 1\,391 & 427 \\ \hline
		\textit{Arxiv COND-MAT} & 521 & 346 \\ \hline 
		\textit{Arxiv HEP-TH} & 12\,530 & 725 \\ \hline 
		\textit{Enron-Email} & 2\,498 & 525 \\ \hline 
		\textit{EU-Email} & 4\,479 & 958 \\ \hline 
		\textit{DBLP} & 4\,620 & 2\,167 \\ \hline 
	\end{tabular}
\end{table}

\begin{figure}[h]
	\centering
	\begin{subfigure}[b]{.49\textwidth}
		\centering
		\includegraphics[width=\textwidth]{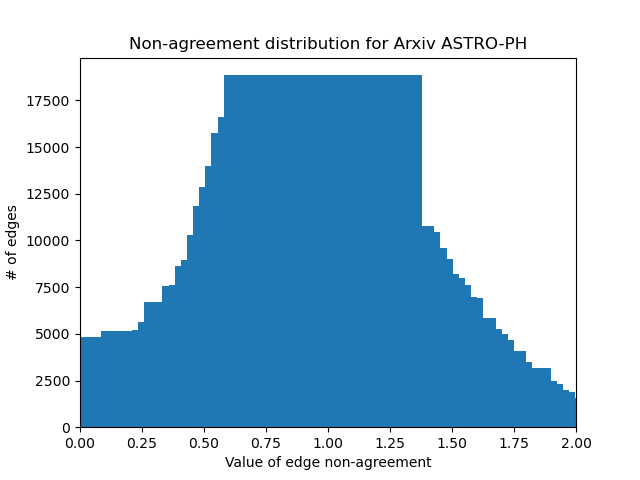}
		\caption{\textsc{NonAgreement}s}
		\label{fig:non:agree:1}
	\end{subfigure}
	\hfill
	\begin{subfigure}[b]{.49\textwidth}
		\centering
		\includegraphics[width=\textwidth]{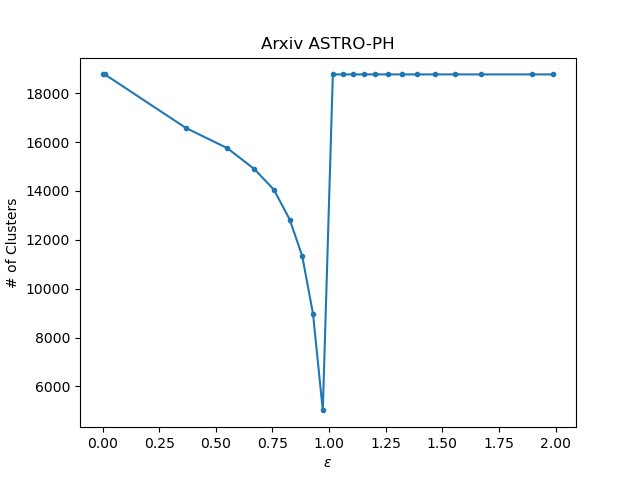}
		\caption{Number of clusters}%
		\label{fig:cluster:1}
	\end{subfigure}
	\\
	\begin{subfigure}[b]{.49\textwidth}
		\centering
		\includegraphics[width=\textwidth]{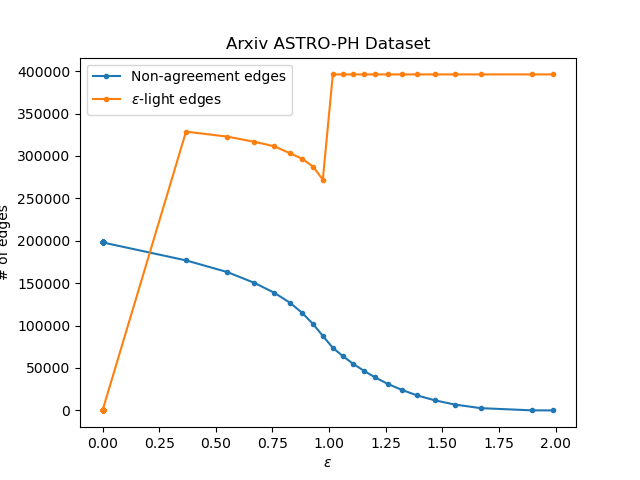}
		\caption{The number of \textsc{NonAgreement} and $\varepsilon$-light edges}%
		\label{fig:non:agree:edges:1}
	\end{subfigure}
	\hfill
	\begin{subfigure}[b]{.49\textwidth}
		\centering
		\includegraphics[width=\textwidth]{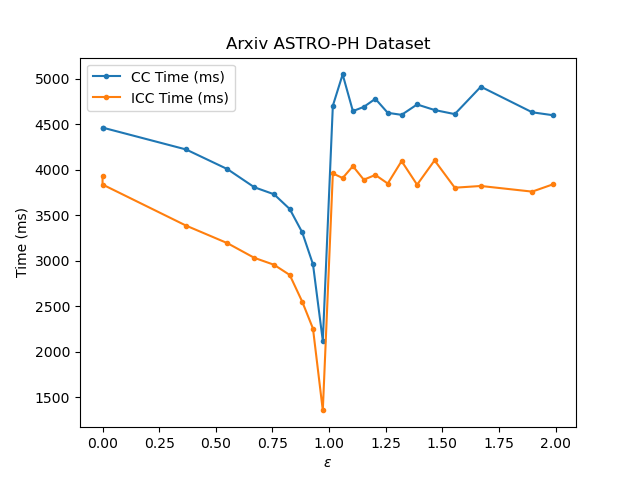}
		\caption{Clustering time in milliseconds}%
		\label{fig:times:1}
	\end{subfigure}
	\caption{The graph Arxiv ASTRO-PH.}
	\label{fig:astro:ph:dataset}
\end{figure}

\begin{figure}[h]
	\centering
	\begin{subfigure}[b]{.49\textwidth}
		\centering
		\includegraphics[width=\textwidth]{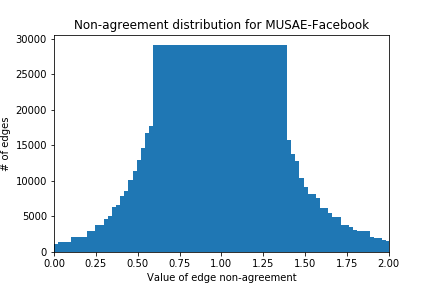}
		\caption{\textsc{NonAgreement}s}%
		\label{fig:non:agree:2}
	\end{subfigure}
	\hfill
	\begin{subfigure}[b]{.49\textwidth}
		\centering
		\includegraphics[width=\textwidth]{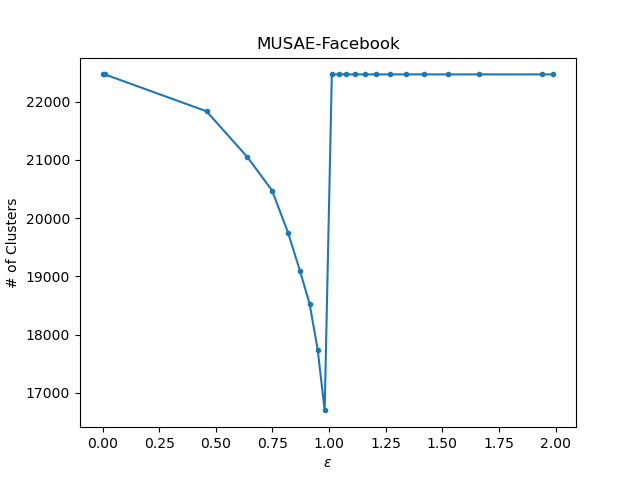}
		\caption{Number of clusters}%
		\label{fig:cluster:2}
	\end{subfigure}
	\\
	\begin{subfigure}[b]{.49\textwidth}
		\centering
		\includegraphics[width=\textwidth]{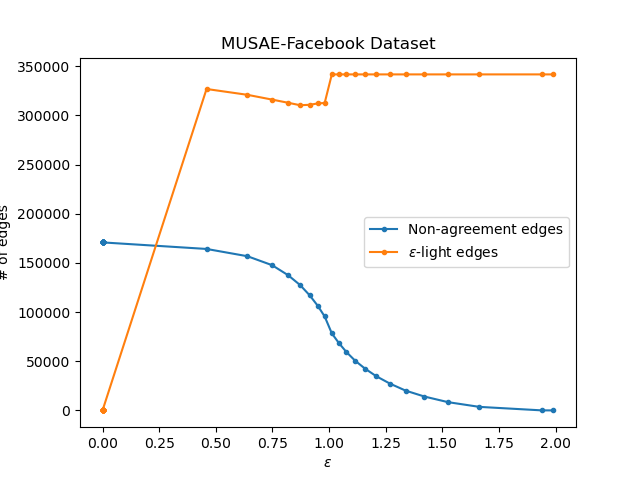}
		\caption{The number of \textsc{NonAgreement} and $\varepsilon$-light edges}%
		\label{fig:non:agree:edges:2}
	\end{subfigure}
	\hfill
	\begin{subfigure}[b]{.49\textwidth}
		\centering
		\includegraphics[width=\textwidth]{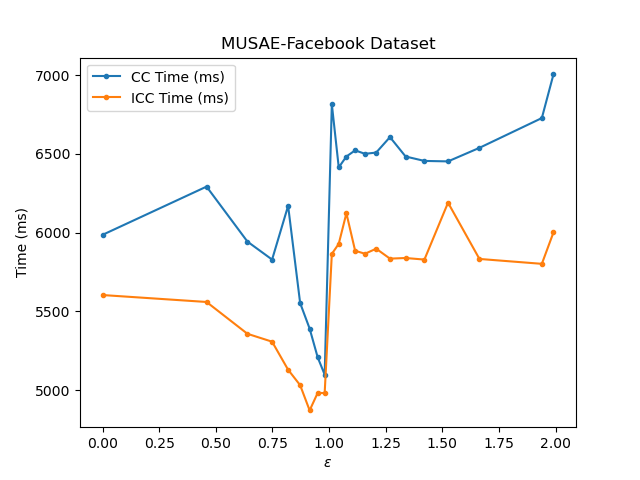}
		\caption{Clustering time in milliseconds}%
		\label{fig:times:2}
	\end{subfigure}
	\caption{The graph MUSAE-Facebook.}
	\label{fig:facebook:dataset}
\end{figure}

\begin{figure}[h]
	\centering
	\begin{subfigure}[b]{.49\textwidth}
		\centering
		\includegraphics[width=\textwidth]{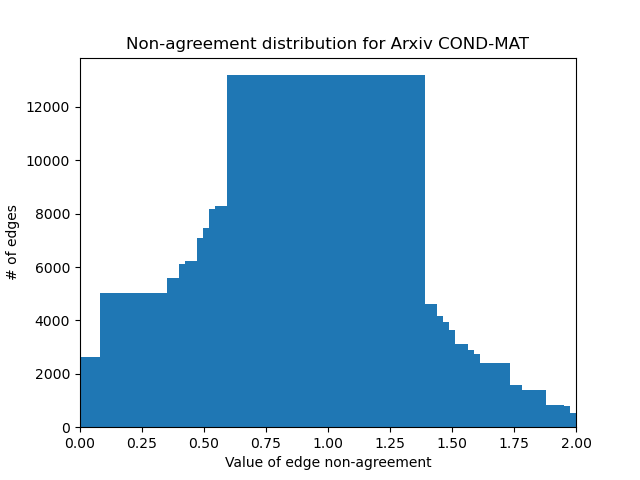}
		\caption{\textsc{NonAgreement}s}%
		\label{fig:non:agree:3}
	\end{subfigure}
	\hfill
	\begin{subfigure}[b]{.49\textwidth}
		\centering
		\includegraphics[width=\textwidth]{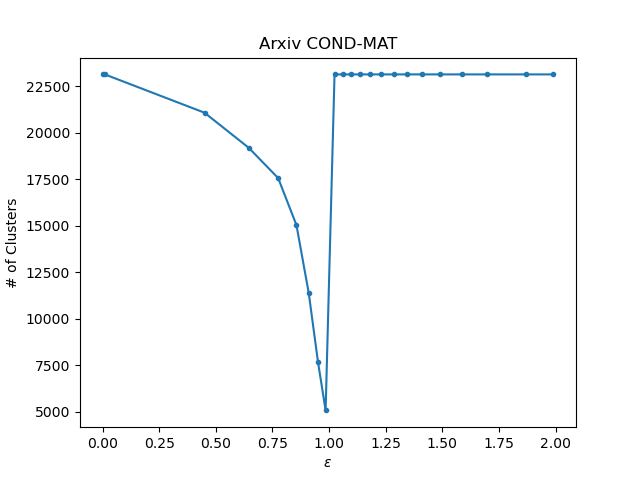}
		\caption{Number of clusters}%
		\label{fig:cluster:3}
	\end{subfigure}
	\\
	\begin{subfigure}[b]{.49\textwidth}
		\centering
		\includegraphics[width=\textwidth]{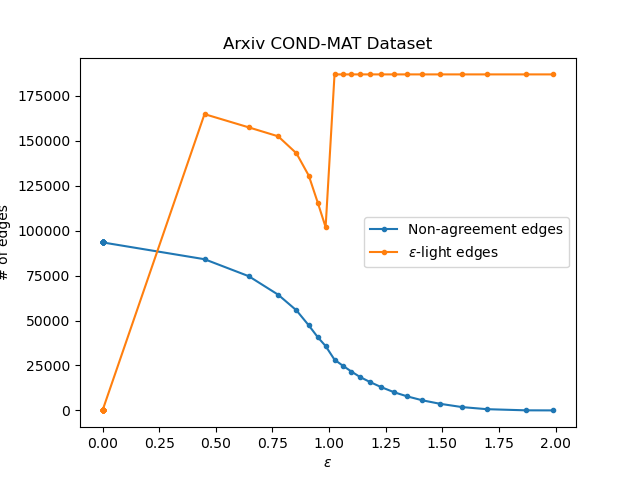}
		\caption{The number of \textsc{NonAgreement} and $\varepsilon$-light edges}%
		\label{fig:non:agree:edges:3}
	\end{subfigure}
	\hfill
	\begin{subfigure}[b]{.49\textwidth}
		\centering
		\includegraphics[width=\textwidth]{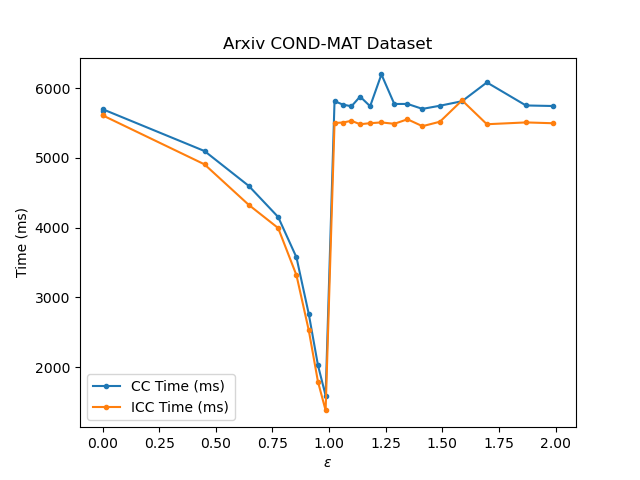}
		\caption{Clustering time in milliseconds}%
		\label{fig:times:3}
	\end{subfigure}
	\caption{The graph Arxiv COND-MAT.}
	\label{fig:cond:mat:dataset}
\end{figure}

\begin{figure}[h]
	\centering
	\begin{subfigure}[b]{.49\textwidth}
		\centering
		\includegraphics[width=\textwidth]{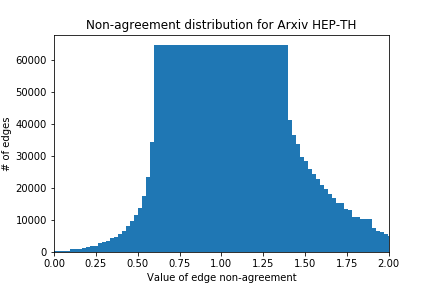}
		\caption{\textsc{NonAgreement}s}%
		\label{fig:non:agree:4}
	\end{subfigure}
	\hfill
	\begin{subfigure}[b]{.49\textwidth}
		\centering
		\includegraphics[width=\textwidth]{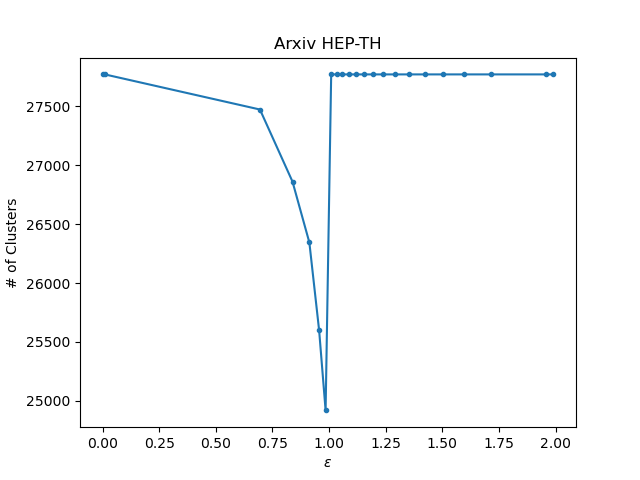}
		\caption{Number of clusters}%
		\label{fig:cluster:4}
	\end{subfigure}
	\\
	\begin{subfigure}[b]{.49\textwidth}
		\centering
		\includegraphics[width=\textwidth]{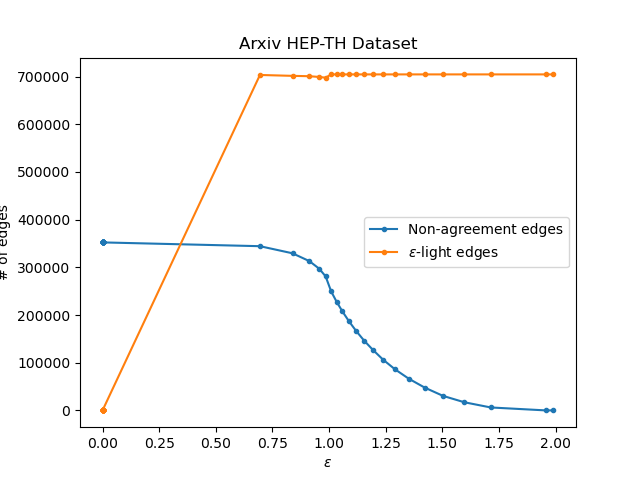}
		\caption{The number of \textsc{NonAgreement} and $\varepsilon$-light edges}%
		\label{fig:non:agree:edges:4}
	\end{subfigure}
	\hfill
	\begin{subfigure}[b]{.49\textwidth}
		\centering
		\includegraphics[width=\textwidth]{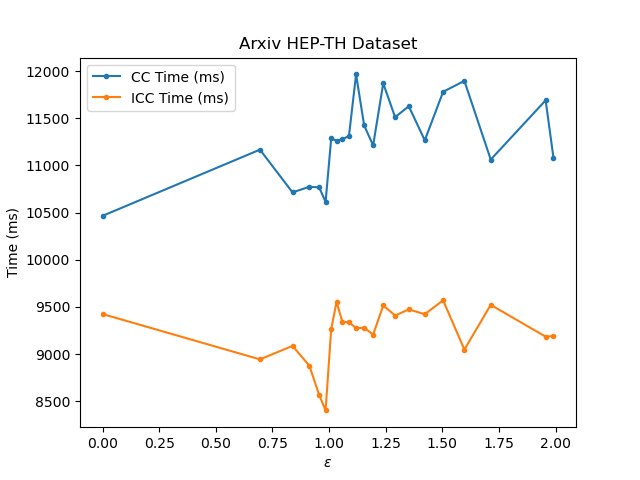}
		\caption{Clustering time in milliseconds}%
		\label{fig:times:4}
	\end{subfigure}
	\caption{The graph Arxiv HEP-TH.}
	\label{fig:hep:th:dataset}
\end{figure}

\begin{figure}[h]
	\centering
	\begin{subfigure}[b]{.49\textwidth}
		\centering
		\includegraphics[width=\textwidth]{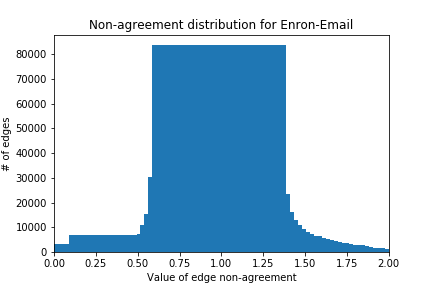}
		\caption{\textsc{NonAgreement}s}%
		\label{fig:non:agree:5}
	\end{subfigure}
	\hfill
	\begin{subfigure}[b]{.49\textwidth}
		\centering
		\includegraphics[width=\textwidth]{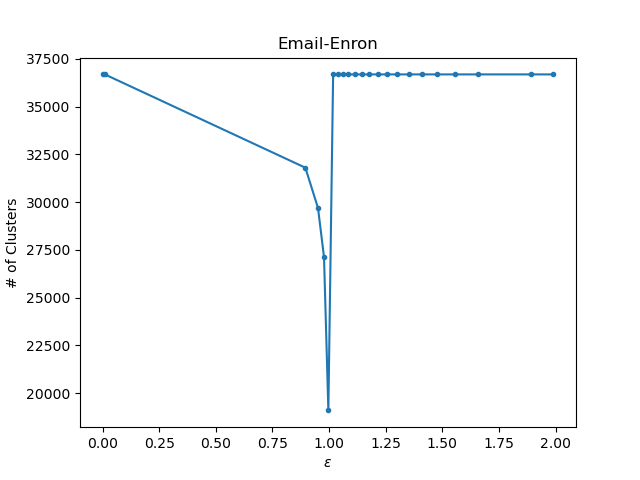}
		\caption{Number of clusters}%
		\label{fig:cluster:5}
	\end{subfigure}
	\\
	\begin{subfigure}[b]{.49\textwidth}
		\centering
		\includegraphics[width=\textwidth]{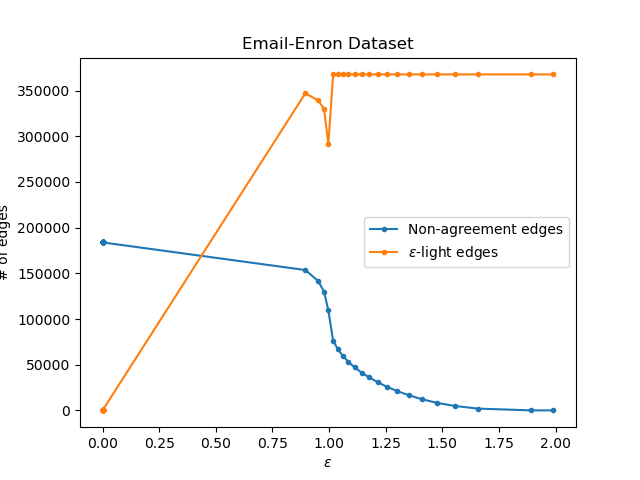}
		\caption{The number of \textsc{NonAgreement} and $\varepsilon$-light edges}%
		\label{fig:non:agree:edges:5}
	\end{subfigure}
	\hfill
	\begin{subfigure}[b]{.49\textwidth}
		\centering
		\includegraphics[width=\textwidth]{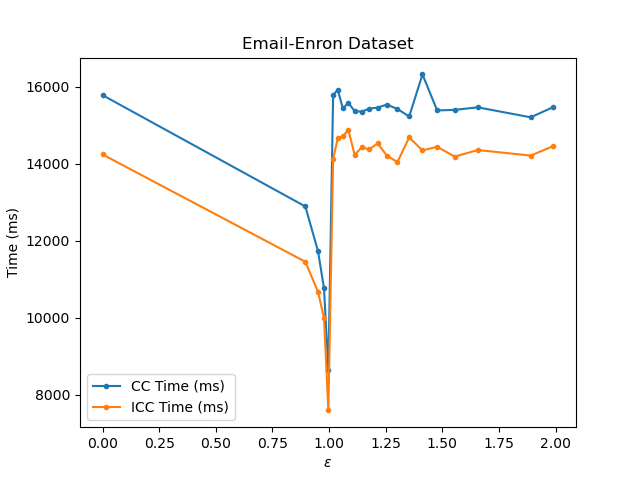}
		\caption{Clustering time in milliseconds}%
		\label{fig:times:5}
	\end{subfigure}
	\caption{The graph Enron-Email.}
	\label{fig:enron:dataset}
\end{figure}

\begin{figure}[h]
	\centering
	\begin{subfigure}[b]{.49\textwidth}
		\centering
		\includegraphics[width=\textwidth]{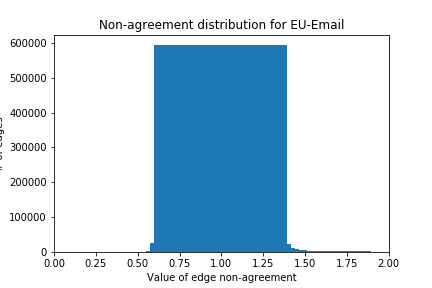}
		\caption{\textsc{NonAgreement}s}%
		\label{fig:non:agree:6}
	\end{subfigure}
	\hfill
	\begin{subfigure}[b]{.49\textwidth}
		\centering
		\includegraphics[width=\textwidth]{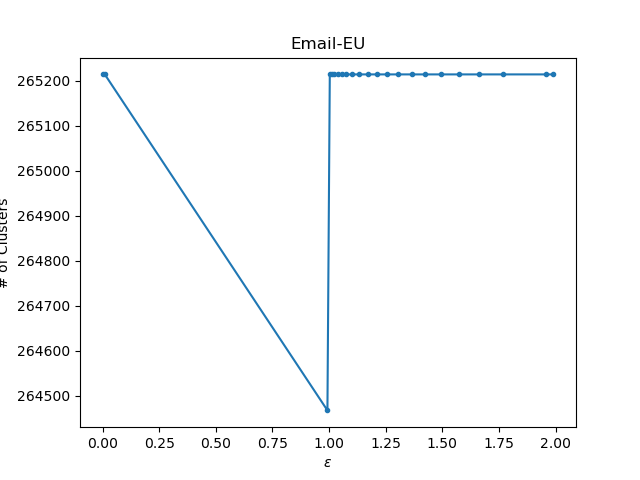}
		\caption{Number of clusters}%
		\label{fig:cluster:6}
	\end{subfigure}
	\\
	\begin{subfigure}[b]{.49\textwidth}
		\centering
		\includegraphics[width=\textwidth]{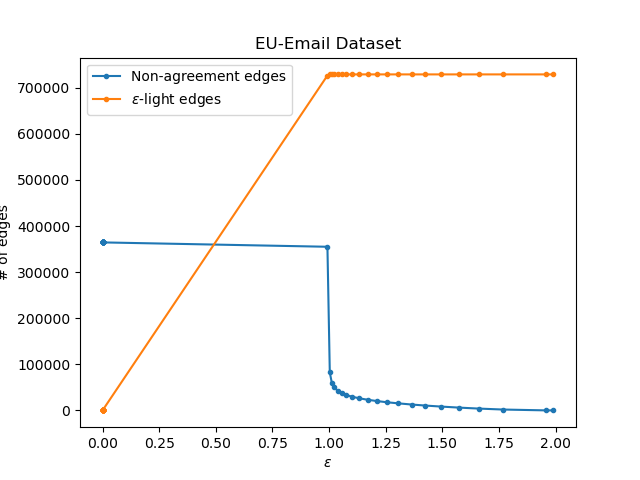}
		\caption{The number of \textsc{NonAgreement} and $\varepsilon$-light edges}%
		\label{fig:non:agree:edges:6}
	\end{subfigure}
	\hfill
	\begin{subfigure}[b]{.49\textwidth}
		\centering
		\includegraphics[width=\textwidth]{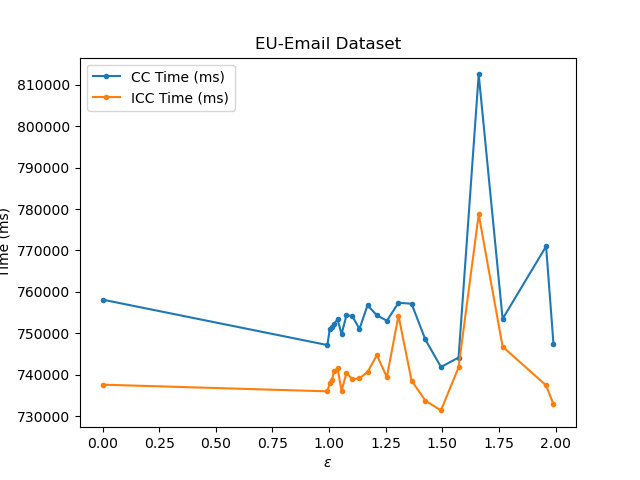}
		\caption{Clustering time in milliseconds}%
		\label{fig:times:6}
	\end{subfigure}
	\caption{The graph EU-Email.}
	\label{fig:eu:dataset}
\end{figure}

\begin{figure}[h]
	\centering
	\begin{subfigure}[b]{.49\textwidth}
		\centering
		\includegraphics[width=\textwidth]{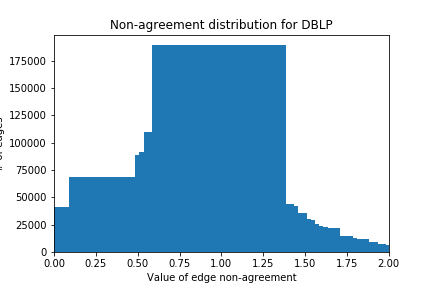}
		\caption{\textsc{NonAgreement}s}%
		\label{fig:non:agree:7}
	\end{subfigure}
	\hfill
	\begin{subfigure}[b]{.49\textwidth}
		\centering
		\includegraphics[width=\textwidth]{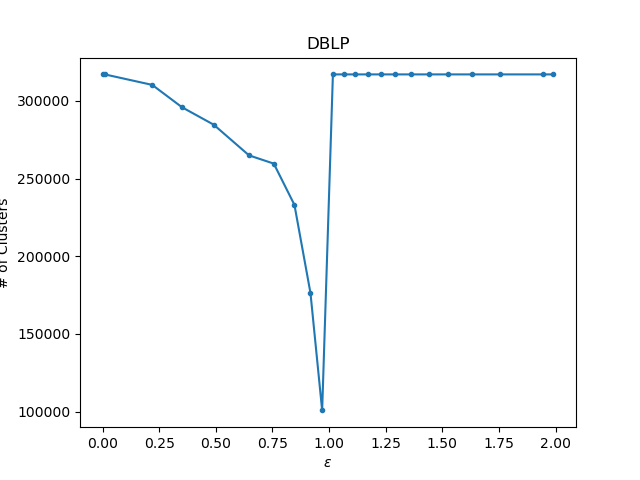}
		\caption{Number of clusters}%
		\label{fig:cluster:7}
	\end{subfigure}
	\\
	\begin{subfigure}[b]{.49\textwidth}
		\centering
		\includegraphics[width=\textwidth]{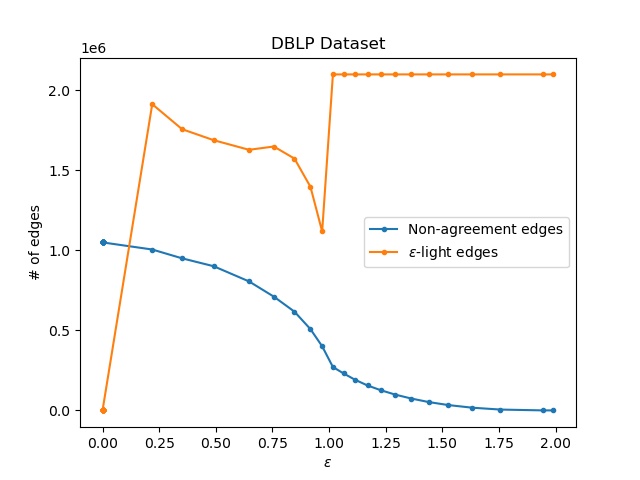}
		\caption{The number of \textsc{NonAgreement} and $\varepsilon$-light edges}%
		\label{fig:non:agree:edges:7}
	\end{subfigure}
	\hfill
	\begin{subfigure}[b]{.49\textwidth}
		\centering
		\includegraphics[width=\textwidth]{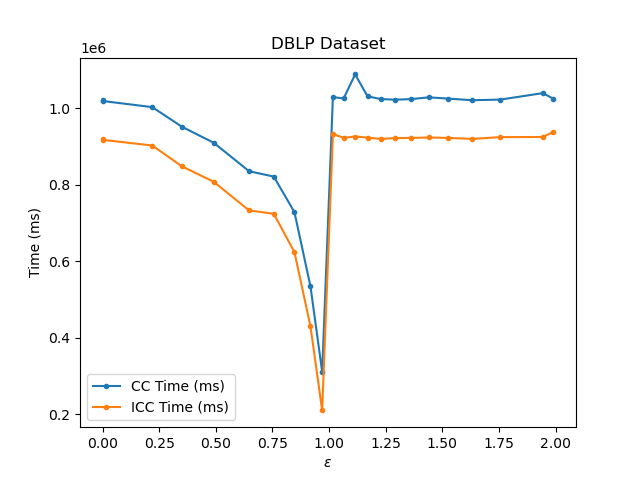}
		\caption{Clustering time in milliseconds}%
		\label{fig:times:7}
	\end{subfigure}
	\caption{The graph DBLP.}
	\label{fig:dblp:dataset}
\end{figure}

%% file: conclusion.tex
In this paper, we proposed a novel indexing structure to decrease the overall running-time of an approximation algorithm for the correlation clustering problem. This structure can be constructed in $\bigoh{m\times\alpha(G)}$ time with $\bigoh{m}$ memory. Then, we can output a correlation clustering for any value of $\varepsilon$ in $\bigoh{m+n}$, compared with
$\bigoh{m\times\left( 2+ \alpha (G) \right)+n}$
time complexity of the ordinary correlation clustering algorithm. Moreover, the proposed index can be efficiently maintained during updates to the underlying graph, including edge sign flip, vertex addition and vertex deletion. The theoretical results are accompanied with practical results in the experiments using seven real world graphs. The experimental results show about \%34 decrease in the running-time of queries. 

A future research direction would be studying this algorithm in parallel frameworks such as Map-Reduce and make it scalable to very Big graphs. Another research direction would be enhancing the approximation guarantee of the algorithm, or devising more efficient algorithms in terms of approximation ratio.